\documentclass[submission,copyright,creativecommons]{eptcs}
\providecommand{\event}{QPL 2018} 
\usepackage{breakurl}             
\usepackage{underscore}           

\usepackage{docmute}
\usepackage{hyperref}
\usepackage[sort,numbers]{natbib}
\usepackage{multirow}
\usepackage{esvect}
\usepackage{amsmath,amsfonts,amssymb,xspace}
\usepackage{amsthm}
\usepackage{tikz}
\usetikzlibrary{fit,shapes}
\usepackage{etoolbox}
\usepackage{datetime}
\usepackage{amssymb}
\newenvironment{pic}[1][]
{\begin{aligned}\begin{tikzpicture}[#1]}
{\end{tikzpicture}\end{aligned}}
\newcommand{\edges}[1][]%
{
}
\theoremstyle{plain}
\newtheorem{theorem}{Theorem}
\newtheorem{lemma}[theorem]{Lemma}
\newtheorem{proposition}[theorem]{Proposition}
\newtheorem{corollary}[theorem]{Corollary}

\theoremstyle{definition}
\newtheorem{definition}[theorem]{Definition} 
\newtheorem{example}[theorem]{Example}

\newtheorem{remark}[theorem]{Remark}

\makeatletter
\def\calign@preamble{%
   &\hfil\strut@
    \setboxz@h{\@lign$\m@th\displaystyle{##}$}%
    \ifmeasuring@\savefieldlength@\fi
    \set@field
    \hfil
    \tabskip\alignsep@
}
\let\cmeasure@\measure@
\patchcmd\cmeasure@{\divide\@tempcntb\tw@}{}{}{}
\patchcmd\cmeasure@{\divide\@tempcntb\tw@}{}{}{}
\patchcmd\cmeasure@{\ifodd\maxfields@
  \global\advance\maxfields@\@ne
  \fi}{}{}{}    
  
\makeatother

\newcommand\tinymatrix[1]
{\left( \hspace{-2pt} \renewcommand\thickspace{\kern2pt} \scriptstyle\begin{smallmatrix} #1 \end{smallmatrix} \hspace{-2pt} \right)}

\newcommand\ignore[1]{}

\DeclareMathOperator{\Tr}{Tr}

\renewcommand\dag{\ensuremath{\dagger}}

\newcommand\C{\ensuremath{\mathbb{C}}}

\usepackage{color}

\def\arraystretch{1.0}
\newcommand\grid[1]{\ensuremath{\def\arraystretch{1.4}\begin{array}{|c|c|c|c|c|c|c|c|c|}\hline#1\\\hline\end{array}}}

\newcommand\I{\ensuremath{\mathbb I}}

\newcommand\super[2]{\stackrel{\makebox[0pt]{\smash{\tiny #1}}}{#2}}

\newcommand\bra[1]{\langle #1|}
\newcommand\ket[1]{{|} #1 \rangle}
\newcommand\braket[2]{\langle #1 | #2 \rangle}
\newcommand\ketbra[2]{|#1 \rangle \hspace{-1pt} \langle #2 |}

\newcounter{jamiecomment}
\newcounter{bencomment}
\newcommand\BMcomm[1]{\ensuremath{{}^{\color{blue}\thebencomment}}\marginpar{\color{blue}\tiny\raggedright \thebencomment: #1}\stepcounter{bencomment}}

\usepackage{amsbsy}
\usepackage[normalem]{ulem}
\allowdisplaybreaks[1]
\usepackage{docmute}
\usepackage{hyperref}
\usepackage{enumitem,multirow}
\usepackage{amsmath,amsfonts,amssymb,mathtools,xspace}
\usepackage{amsthm}
\usepackage{tikz}
\usepackage{array}
\usetikzlibrary{decorations.pathreplacing}
\usetikzlibrary{decorations.markings}
\usetikzlibrary{calc}
\usepackage{etoolbox}
\usepackage{datetime}

\tikzstyle{dc}   = [circle, minimum width=8pt, draw, inner sep=0pt, path picture={\draw (path picture bounding box.south east) -- (path picture bounding box.north west) (path picture bounding box.south west) -- (path picture bounding box.north east);}]
\tikzstyle{dp}   = [circle, minimum width=8pt, draw, inner sep=0pt, path picture={\draw (path picture bounding box.west) -- (path picture bounding box.north) (path picture bounding box.south west) -- (path picture bounding box.north) (path picture bounding box.south west) -- (path picture bounding box.north east) (path picture bounding box.north east) -- (path picture bounding box.south) (path picture bounding box.south) -- (path picture bounding box.east);}]
\tikzstyle{ls} = [dc,scale=0.65]
\tikzstyle{ls'} = [dp,scale=1.3]
\tikzstyle{dpt}  = [circle, minimum width=8pt, draw, inner sep=0pt, path picture={\draw (path picture bounding box.south) -- (path picture bounding box.north) (path picture bounding box.west) -- (path picture bounding box.east);}]
\tikzstyle{agg} = [dpt,scale=0.65]

\renewenvironment{pic}[1][]
{\begin{aligned}\begin{tikzpicture}[#1]}
{\end{tikzpicture}\end{aligned}}

\makeatletter
\def\calign@preamble{%
   &\hfil\strut@
    \setboxz@h{\@lign$\m@th\displaystyle{##}$}%
    \ifmeasuring@\savefieldlength@\fi
    \set@field
    \hfil
    \tabskip\alignsep@
}
\let\cmeasure@\measure@
\patchcmd\cmeasure@{\divide\@tempcntb\tw@}{}{}{}
\patchcmd\cmeasure@{\divide\@tempcntb\tw@}{}{}{}
\patchcmd\cmeasure@{\ifodd\maxfields@
  \global\advance\maxfields@\@ne
  \fi}{}{}{}    

\makeatother


\pgfdeclarelayer{foreground}
\pgfdeclarelayer{background}
\pgfsetlayers{main,foreground,background}

\makeatletter
\pgfkeys{%
  /tikz/on layer/.code={
    \pgfonlayer{#1}\begingroup
    \aftergroup\endpgfonlayer
    \aftergroup\endgroup
  },
  /tikz/node on layer/.code={
    \gdef\node@@on@layer{%
      \setbox\tikz@tempbox=\hbox\bgroup\pgfonlayer{#1}\unhbox\tikz@tempbox\endpgfonlayer\egroup}
    \aftergroup\node@on@layer
  },
  /tikz/end node on layer/.code={
    \endpgfonlayer\endgroup\endgroup
  }
}
\def\node@on@layer{\aftergroup\node@@on@layer}
\makeatother\def\thickness{0.7pt}
\tikzstyle{oldmorphism}=[minimum width=30pt, minimum height=16pt, draw, font=\small, inner sep=0pt, fill=white, line width=\thickness]
\tikzstyle{cross}=[preaction={draw=white, -, line width=10pt}]
\tikzstyle{braid}=[double=black, line width=3*\thickness, double distance=\thickness, white]
\tikzstyle{string}=[line width=\thickness]
\tikzstyle{scalar}=[circle, inner sep=0pt, minimum width=15pt, draw, line width=\thickness]
\tikzstyle{dot}=[circle, draw=black, fill=black!25, inner sep=.5ex, line width=\thickness, node on layer=foreground]
\tikzstyle{blackdot}=[circle, draw=black, fill=black, inner sep=.5ex, line width=\thickness, node on layer=foreground]
\tikzstyle{whitedot}=[circle, draw=black, fill=white, inner sep=.5ex, line width=\thickness, node on layer=foreground]
\tikzstyle{reddot}=[circle, draw=black, fill=red, inner sep=.5ex, line width=\thickness, node on layer=foreground]
\tikzstyle{bluedot}=[circle, draw=black, fill=blue, inner sep=.5ex, line width=\thickness, node on layer=foreground]
\tikzstyle{yellowdot}=[circle, draw=black, fill=yellow, inner sep=.5ex, line width=\thickness, node on layer=foreground]
\tikzstyle{greendot}=[circle, draw=black, fill=green, inner sep=.5ex, line width=\thickness, node on layer=foreground]
\tikzstyle{lsdot}=[circle, draw=black, fill=white, inner sep=.5ex, line width=\thickness, node on layer=foreground, path picture={\draw (path picture bounding box.south east) -- (path picture bounding box.north west) (path picture bounding box.south west) -- (path picture bounding box.north east);}]
\tikzstyle{lssdot}=[circle, draw=black, fill=white, inner sep=.5ex, line width=\thickness, node on layer=foreground, path picture={\draw (path picture bounding box.south) -- (path picture bounding box.north) (path picture bounding box.west) -- (path picture bounding box.east);}]
\tikzstyle{mixedmorphism}=[morphism, minimum width=30pt, minimum height=16pt, draw, font=\small, inner sep=0pt, fill=white, line width=\thickness,rounded corners=1ex]
\tikzstyle{thick}=[line width=\thickness]
\tikzstyle{tiny}=[font=\tiny]

\tikzset{arrow/.style={decoration={
    markings,
    mark=at position #1 with \arrow{thickarrow}},
    postaction=decorate}
}
\tikzset{reverse arrow/.style={decoration={
    markings,
    mark=at position #1 with \arrow{reversethickarrow}},
    postaction=decorate}
}

\newif\ifblack\pgfkeys{/tikz/black/.is if=black}
\newif\ifwedge\pgfkeys{/tikz/wedge/.is if=wedge}
\newif\ifvflip\pgfkeys{/tikz/vflip/.is if=vflip}
\newif\ifhflip\pgfkeys{/tikz/hflip/.is if=hflip}
\newif\ifhvflip\pgfkeys{/tikz/hvflip/.is if=hvflip}
\newif\ifconnectnw\pgfkeys{/tikz/connect nw/.is if=connectnw}
\newif\ifconnectne\pgfkeys{/tikz/connect ne/.is if=connectne}
\newif\ifconnectsw\pgfkeys{/tikz/connect sw/.is if=connectsw}
\newif\ifconnectse\pgfkeys{/tikz/connect se/.is if=connectse}
\newif\ifconnectn\pgfkeys{/tikz/connect n/.is if=connectn}
\newif\ifconnects\pgfkeys{/tikz/connect s/.is if=connects}
\newif\ifconnectnwf\pgfkeys{/tikz/connect nw >/.is if=connectnwf}
\newif\ifconnectnef\pgfkeys{/tikz/connect ne >/.is if=connectnef}
\newif\ifconnectswf\pgfkeys{/tikz/connect sw >/.is if=connectswf}
\newif\ifconnectsef\pgfkeys{/tikz/connect se >/.is if=connectsef}
\newif\ifconnectnf\pgfkeys{/tikz/connect n >/.is if=connectnf}
\newif\ifconnectsf\pgfkeys{/tikz/connect s >/.is if=connectsf}
\newif\ifconnectnwr\pgfkeys{/tikz/connect nw </.is if=connectnwr}
\newif\ifconnectner\pgfkeys{/tikz/connect ne </.is if=connectner}
\newif\ifconnectswr\pgfkeys{/tikz/connect sw </.is if=connectswr}
\newif\ifconnectser\pgfkeys{/tikz/connect se </.is if=connectser}
\newif\ifconnectnr\pgfkeys{/tikz/connect n </.is if=connectnr}
\newif\ifconnectsr\pgfkeys{/tikz/connect s </.is if=connectsr}
\tikzset{keylengthnw/.initial=\connectheight}
\tikzset{keylengthn/.initial =\connectheight}
\tikzset{keylengthne/.initial=\connectheight}
\tikzset{keylengthsw/.initial=\connectheight}
\tikzset{keylengths/.initial =\connectheight}
\tikzset{keylengthse/.initial=\connectheight}
\tikzset{connect nw length/.style={connect nw=true, keylengthnw={#1}}}
\tikzset{connect n length/.style ={connect n =true, keylengthn ={#1}}}
\tikzset{connect ne length/.style={connect ne=true, keylengthne={#1}}}
\tikzset{connect sw length/.style={connect sw=true, keylengthsw={#1}}}
\tikzset{connect s length/.style ={connect s =true, keylengths ={#1}}}
\tikzset{connect se length/.style={connect se=true, keylengthse={#1}}}
\tikzset{connect nw < length/.style={connect nw <=true, keylengthnw={#1}}}
\tikzset{connect n < length/.style ={connect n <=true,  keylengthn ={#1}}}
\tikzset{connect ne < length/.style={connect ne <=true, keylengthne={#1}}}
\tikzset{connect sw < length/.style={connect sw <=true, keylengthnw={#1}}}
\tikzset{connect s < length/.style ={connect s <=true,  keylengths ={#1}}}
\tikzset{connect se < length/.style={connect se <=true, keylengthse={#1}}}
\tikzset{connect nw > length/.style={connect nw >=true, keylengthnw={#1}}}
\tikzset{connect n > length/.style ={connect n >=true,  keylengthn ={#1}}}
\tikzset{connect ne > length/.style={connect ne >=true, keylengthne={#1}}}
\tikzset{connect sw > length/.style={connect sw >=true, keylengthsw={#1}}}
\tikzset{connect s > length/.style ={connect s >=true,  keylengths ={#1}}}
\tikzset{connect se > length/.style={connect se >=true, keylengthse={#1}}}

\newlength\morphismheight
\newlength\minimummorphismwidth
\newlength\stateheight
\newlength\minimumstatewidth
\newlength\connectheight
\tikzset{width/.initial=\minimummorphismwidth}

\makeatletter
\pgfarrowsdeclare{thickarrow}{thickarrow}
{
  \pgfutil@tempdima=-0.84pt%
  \advance\pgfutil@tempdima by-1.3\pgflinewidth%
  \pgfutil@tempdimb=-1.7pt%
  \advance\pgfutil@tempdimb by.625\pgflinewidth%
  \pgfarrowsleftextend{+\pgfutil@tempdima}
  \pgfarrowsrightextend{+\pgfutil@tempdimb}
}
{
  \pgfmathparse{\pgfgetarrowoptions{thickarrow}}%
  \pgfsetlinewidth{1.25 pt}
  \pgfutil@tempdima=0.28pt%
  \advance\pgfutil@tempdima by.3\pgflinewidth%
  \pgfsetlinewidth{0.8\pgflinewidth}
  \pgfsetdash{}{+0pt}
  \pgfsetroundcap
  \pgfsetroundjoin
  \pgfpathmoveto{\pgfqpoint{-3\pgfutil@tempdima}{4\pgfutil@tempdima}}
  \pgfpathcurveto
  {\pgfqpoint{-2.75\pgfutil@tempdima}{2.5\pgfutil@tempdima}}
  {\pgfqpoint{0pt}{0.25\pgfutil@tempdima}}
  {\pgfqpoint{0.75\pgfutil@tempdima}{0pt}}
  \pgfpathcurveto
  {\pgfqpoint{0pt}{-0.25\pgfutil@tempdima}}
  {\pgfqpoint{-2.75\pgfutil@tempdima}{-2.5\pgfutil@tempdima}}
  {\pgfqpoint{-3\pgfutil@tempdima}{-4\pgfutil@tempdima}}
  \pgfusepathqstroke
}
\pgfarrowsdeclare{reversethickarrow}{reversethickarrow}
{
  \pgfutil@tempdima=-0.84pt%
  \advance\pgfutil@tempdima by-1.3\pgflinewidth%
  \pgfutil@tempdimb=0.2pt%
  \advance\pgfutil@tempdimb by.625\pgflinewidth%
  \pgfarrowsleftextend{+\pgfutil@tempdima}
  \pgfarrowsrightextend{+\pgfutil@tempdimb}
}
{
  \pgftransformxscale{-1}
  \pgfmathparse{\pgfgetarrowoptions{thickarrow}}%
  \ifpgfmathunitsdeclared%
    \pgfmathparse{\pgfmathresult pt}%
  \else%
    \pgfmathparse{\pgfmathresult*\pgflinewidth}%
  \fi%
  \let\thickness=\pgfmathresult
  \pgfsetlinewidth{1.25 pt}
  \pgfutil@tempdima=0.28pt%
  \advance\pgfutil@tempdima by.3\pgflinewidth%
  \pgfsetlinewidth{0.8\pgflinewidth}
  \pgfsetdash{}{+0pt}
  \pgfsetroundcap
  \pgfsetroundjoin
  \pgfpathmoveto{\pgfqpoint{-3\pgfutil@tempdima}{4\pgfutil@tempdima}}
  \pgfpathcurveto
  {\pgfqpoint{-2.75\pgfutil@tempdima}{2.5\pgfutil@tempdima}}
  {\pgfqpoint{0pt}{0.25\pgfutil@tempdima}}
  {\pgfqpoint{0.75\pgfutil@tempdima}{0pt}}
  \pgfpathcurveto
  {\pgfqpoint{0pt}{-0.25\pgfutil@tempdima}}
  {\pgfqpoint{-2.75\pgfutil@tempdima}{-2.5\pgfutil@tempdima}}
  {\pgfqpoint{-3\pgfutil@tempdima}{-4\pgfutil@tempdima}}
  \pgfusepathqstroke
}
\makeatother

\makeatletter
\pgfdeclareshape{ground}
{
    \savedanchor\centerpoint
    {
        \pgf@x=0pt
        \pgf@y=0pt
    }
    \anchor{center}{\centerpoint}
    \anchorborder{\centerpoint}
    \anchor{north}
    {
        \pgf@x=0pt
        \pgf@y=0.5*0.33*\stateheight
    }
    \anchor{south}
    {
        \pgf@x=0pt
        \pgf@y=0pt
    }
    \saveddimen\overallwidth
    {
        \pgfkeysgetvalue{/pgf/minimum width}{\minwidth}
        \pgf@x=\minimumstatewidth
        \ifdim\pgf@x<\minwidth
            \pgf@x=\minwidth
        \fi
    }
    \backgroundpath
    {
        \begin{pgfonlayer}{foreground}
        \pgfsetstrokecolor{black}
        \pgfsetcolor{gray}
        \pgfsetlinewidth{1.25pt}
        \ifhflip
            \pgftransformyscale{-1}
        \fi
        \pgftransformscale{0.5}
        \pgfpathmoveto{\pgfpoint{-0.5*\overallwidth}{0}}
        \pgfpathlineto{\pgfpoint{0.5*\overallwidth}{0}}
        \pgfpathmoveto{\pgfpoint{-0.33*\overallwidth}{0.33*\stateheight}}
        \pgfpathlineto{\pgfpoint{0.33*\overallwidth}{0.33*\stateheight}}
        \pgfpathmoveto{\pgfpoint{-0.16*\overallwidth}{0.66*\stateheight}}
        \pgfpathlineto{\pgfpoint{0.16*\overallwidth}{0.66*\stateheight}}
        \pgfpathmoveto{\pgfpoint{-0.02*\overallwidth}{\stateheight}}
        \pgfpathlineto{\pgfpoint{0.02*\overallwidth}{\stateheight}}
        \end{pgfonlayer}
    }
}
\tikzset{forward arrow style/.style={every to/.style, decoration={
    markings,
    mark=at position 0.5 with \arrow{thickarrow}},
    postaction=decorate}}
\tikzset{reverse arrow style/.style={every to/.style, decoration={
    markings,
    mark=at position 0.5 with \arrow{reversethickarrow}},
    postaction=decorate}}
\pgfdeclareshape{morphism}
{
    \savedanchor\centerpoint
    {
        \pgf@x=0pt
        \pgf@y=0pt
    }
    \anchor{center}{\centerpoint}
    \anchorborder{\centerpoint}
    \saveddimen\savedlengthnw
    {
        \pgfkeysgetvalue{/tikz/keylengthnw}{\len}
        \pgf@x=\len
    }
    \saveddimen\savedlengthn
    {
        \pgfkeysgetvalue{/tikz/keylengthn}{\len}
        \pgf@x=\len
    }
    \saveddimen\savedlengthne
    {
        \pgfkeysgetvalue{/tikz/keylengthne}{\len}
        \pgf@x=\len
    }
    \saveddimen\savedlengthsw
    {
        \pgfkeysgetvalue{/tikz/keylengthsw}{\len}
        \pgf@x=\len
    }
    \saveddimen\savedlengths
    {
        \pgfkeysgetvalue{/tikz/keylengths}{\len}
        \pgf@x=\len
    }
    \saveddimen\savedlengthse
    {
        \pgfkeysgetvalue{/tikz/keylengthse}{\len}
        \pgf@x=\len
    }
    \saveddimen\overallwidth
    {
        \pgfkeysgetvalue{/tikz/width}{\minwidth}
        \pgf@x=\wd\pgfnodeparttextbox
        \ifdim\pgf@x<\minwidth
            \pgf@x=\minwidth
        \fi
    }
    \savedanchor{\upperrightcorner}
    {
        \pgf@y=.5\ht\pgfnodeparttextbox
        \advance\pgf@y by -.5\dp\pgfnodeparttextbox
        \pgf@x=.5\wd\pgfnodeparttextbox
    }
    \anchor{north}
    {
        \pgf@x=0pt
        \pgf@y=0.5\morphismheight
    }
    \anchor{north east}
    {
        \pgf@x=\overallwidth
        \multiply \pgf@x by 2
        \divide \pgf@x by 5
        \pgf@y=0.5\morphismheight
    }
    \anchor{east}
    {
        \pgf@x=\overallwidth
        \divide \pgf@x by 2
        \advance \pgf@x by 5pt
        \pgf@y=0pt
    }
    \anchor{west}
    {
        \pgf@x=-\overallwidth
        \divide \pgf@x by 2
        \advance \pgf@x by -5pt
        \pgf@y=0pt
    }
    \anchor{north west}
    {
        \pgf@x=-\overallwidth
        \multiply \pgf@x by 2
        \divide \pgf@x by 5
        \pgf@y=0.5\morphismheight
    }
    \anchor{connect nw}
    {
        \pgf@x=-\overallwidth
        \multiply \pgf@x by 2
        \divide \pgf@x by 5
        \pgf@y=0.5\morphismheight
        \advance\pgf@y by \savedlengthnw
    }
    \anchor{connect ne}
    {
        \pgf@x=\overallwidth
        \multiply \pgf@x by 2
        \divide \pgf@x by 5
        \pgf@y=0.5\morphismheight
        \advance\pgf@y by \savedlengthne
    }
    \anchor{connect sw}
    {
        \pgf@x=-\overallwidth
        \multiply \pgf@x by 2
        \divide \pgf@x by 5
        \pgf@y=-0.5\morphismheight
        \advance\pgf@y by -\savedlengthsw
    }
    \anchor{connect se}
    {
        \pgf@x=\overallwidth
        \multiply \pgf@x by 2
        \divide \pgf@x by 5
        \pgf@y=-0.5\morphismheight
        \advance\pgf@y by -\savedlengthse
    }
    \anchor{connect n}
    {
        \pgf@x=0pt
        \pgf@y=0.5\morphismheight
        \advance\pgf@y by \savedlengthn
    }
    \anchor{connect s}
    {
        \pgf@x=0pt
        \pgf@y=-0.5\morphismheight
        \advance\pgf@y by -\savedlengths
    }
    \anchor{south east}
    {
        \pgf@x=\overallwidth
        \multiply \pgf@x by 2
        \divide \pgf@x by 5
        \pgf@y=-0.5\morphismheight
    }
    \anchor{south west}
    {
        \pgf@x=-\overallwidth
        \multiply \pgf@x by 2
        \divide \pgf@x by 5
        \pgf@y=-0.5\morphismheight
    }
    \anchor{south}
    {
        \pgf@x=0pt
        \pgf@y=-0.5\morphismheight
    }
    \anchor{text}
    {
        \upperrightcorner
        \pgf@x=-\pgf@x
        \pgf@y=-\pgf@y
    }
    \backgroundpath
    {
        \pgfsetstrokecolor{black}
        \pgfsetlinewidth{\thickness}
        \begin{scope}
                \ifhflip
                    \pgftransformyscale{-1}
                \fi
                \ifvflip
                    \pgftransformxscale{-1}
                \fi
                \ifhvflip
                    \pgftransformxscale{-1}
                    \pgftransformyscale{-1}
                \fi
                \pgfpathmoveto{\pgfpoint
                    {-0.5*\overallwidth-5pt}
                    {0.5*\morphismheight}}
                \pgfpathlineto{\pgfpoint
                    {0.5*\overallwidth+5pt}
                    {0.5*\morphismheight}}
                \ifwedge
                    \pgfpathlineto{\pgfpoint
                        {0.5*\overallwidth + 15pt}
                        {-0.5*\morphismheight}}
                \else
                    \pgfpathlineto{\pgfpoint
                        {0.5*\overallwidth + 5pt}
                        {-0.5*\morphismheight}}
                \fi
                \pgfpathlineto{\pgfpoint
                    {-0.5*\overallwidth-5pt}
                    {-0.5*\morphismheight}}
                \pgfpathclose
                \pgfusepath{stroke}
        \end{scope}
        \ifconnectnw
            \pgfpathmoveto{\pgfpoint
                {-0.4*\overallwidth}
                {0.5*\morphismheight}}
            \pgfpathlineto{\pgfpoint
                {-0.4*\overallwidth}
                {0.5*\morphismheight+\savedlengthnw}}
            \pgfusepath{stroke}
        \fi
        \ifconnectne
            \pgfpathmoveto{\pgfpoint
                {0.4*\overallwidth}
                {0.5*\morphismheight}}
            \pgfpathlineto{\pgfpoint
                {0.4*\overallwidth}
                {0.5*\morphismheight+\savedlengthne}}
            \pgfusepath{stroke}
        \fi
        \ifconnectsw
            \pgfpathmoveto{\pgfpoint
                {-0.4*\overallwidth}
                {-0.5*\morphismheight}}
            \pgfpathlineto{\pgfpoint
                {-0.4*\overallwidth}
                {-0.5*\morphismheight-\savedlengthsw}}
            \pgfusepath{stroke}
        \fi
        \ifconnectse
            \pgfpathmoveto{\pgfpoint
                {0.4*\overallwidth}
                {-0.5*\morphismheight}}
            \pgfpathlineto{\pgfpoint
                {0.4*\overallwidth}
                {-0.5*\morphismheight-\savedlengthse}}
            \pgfusepath{stroke}
        \fi
        \ifconnectn
            \pgfpathmoveto{\pgfpoint
                {0pt}
                {0.5*\morphismheight}}
            \pgfpathlineto{\pgfpoint
                {0pt}
                {0.5*\morphismheight+\savedlengthn}}
            \pgfusepath{stroke}
        \fi
        \ifconnects
            \pgfpathmoveto{\pgfpoint
                {0pt}
                {-0.5*\morphismheight}}
            \pgfpathlineto{\pgfpoint
                {0pt}
                {-0.5*\morphismheight-\savedlengths}}
            \pgfusepath{stroke}
        \fi
        \ifconnectnwf
            \draw [forward arrow style] (-0.4*\overallwidth,0.5*\morphismheight)
                to (-0.4*\overallwidth,0.5*\morphismheight+\savedlengthnw);
        \fi
        \ifconnectnef
            \draw [forward arrow style] (0.4*\overallwidth,0.5*\morphismheight)
                to (0.4*\overallwidth,0.5*\morphismheight+\savedlengthne);
        \fi
        \ifconnectswf
            \draw [forward arrow style] (-0.4*\overallwidth,-0.5*\morphismheight-\savedlengthsw)
                to (-0.4*\overallwidth,-0.5*\morphismheight);
        \fi
        \ifconnectsef
            \draw [forward arrow style] (0.4*\overallwidth,-0.5*\morphismheight-\savedlengthse)
                to (0.4*\overallwidth,-0.5*\morphismheight);
        \fi
        \ifconnectnf
            \draw [forward arrow style] (0,0.5*\morphismheight)
                to (0,0.5*\morphismheight+\savedlengthn);
        \fi
        \ifconnectsf
            \draw [forward arrow style] (0,-0.5*\morphismheight-\savedlengths)
                to (0,-0.5*\morphismheight);
        \fi
        \ifconnectnwr
            \draw [reverse arrow style] (-0.4*\overallwidth,0.5*\morphismheight)
                to (-0.4*\overallwidth,0.5*\morphismheight+\savedlengthnw);
        \fi
        \ifconnectner
            \draw [reverse arrow style] (0.4*\overallwidth,0.5*\morphismheight)
                to (0.4*\overallwidth,0.5*\morphismheight+\savedlengthne);
        \fi
        \ifconnectswr
            \draw [reverse arrow style] (-0.4*\overallwidth,-0.5*\morphismheight-\savedlengthsw)
                to (-0.4*\overallwidth,-0.5*\morphismheight);
        \fi
        \ifconnectser
            \draw [reverse arrow style] (0.4*\overallwidth,-0.5*\morphismheight-\savedlengthse)
                to (0.4*\overallwidth,-0.5*\morphismheight);
        \fi
        \ifconnectnr
            \draw [reverse arrow style] (0,0.5*\morphismheight)
                to (0,0.5*\morphismheight+\savedlengthn);
        \fi
        \ifconnectsr
            \draw [reverse arrow style] (0,-0.5*\morphismheight-\savedlengths)
                to (0,-0.5*\morphismheight);
        \fi
    }
}
\pgfdeclareshape{swish right}
{
    \savedanchor\centerpoint
    {
        \pgf@x=0pt
        \pgf@y=0pt
    }
    \anchor{center}{\centerpoint}
    \anchorborder{\centerpoint}
    \anchor{north}
    {
        \pgf@x=\minimummorphismwidth
        \divide\pgf@x by 5
        \pgf@y=\morphismheight
        \divide\pgf@y by 2
        \advance\pgf@y by \connectheight
    }
    \anchor{south}
    {
        \pgf@x=-\minimummorphismwidth
        \divide\pgf@x by 5
        \pgf@y=-\morphismheight
        \divide\pgf@y by 2
        \advance\pgf@y by -\connectheight
    }
    \backgroundpath
    {
        \pgfsetstrokecolor{black}
        \pgfsetlinewidth{\thickness}
        \pgfpathmoveto{\pgfpoint
            {-0.2*\minimummorphismwidth}
            {-0.5*\morphismheight-\connectheight}}
        \pgfpathcurveto
            {\pgfpoint{-0.2*\minimummorphismwidth}{0pt}}
            {\pgfpoint{0.2*\minimummorphismwidth}{0pt}}
            {\pgfpoint
                {0.2*\minimummorphismwidth}
                {0.5*\morphismheight+\connectheight}}
        \pgfusepath{stroke}
    }
}
\pgfdeclareshape{swish left}
{
    \savedanchor\centerpoint
    {
        \pgf@x=0pt
        \pgf@y=0pt
    }
    \anchor{center}{\centerpoint}
    \anchorborder{\centerpoint}
    \anchor{north}
    {
        \pgf@x=-\minimummorphismwidth
        \divide\pgf@x by 5
        \pgf@y=\morphismheight
        \divide\pgf@y by 2
        \advance\pgf@y by \connectheight
    }
    \anchor{south}
    {
        \pgf@x=\minimummorphismwidth
        \divide\pgf@x by 5
        \pgf@y=-\morphismheight
        \divide\pgf@y by 2
        \advance\pgf@y by -\connectheight
    }
    \backgroundpath
    {
        \pgfsetstrokecolor{black}
        \pgfsetlinewidth{\thickness}
        \pgfpathmoveto{\pgfpoint
            {0.2*\minimummorphismwidth}
            {-0.5*\morphismheight-\connectheight}}
        \pgfpathcurveto
            {\pgfpoint{0.2*\minimummorphismwidth}{0pt}}
            {\pgfpoint{-0.2*\minimummorphismwidth}{0pt}}
            {\pgfpoint
                {-0.2*\minimummorphismwidth}
                {0.5*\morphismheight+\connectheight}}
        \pgfusepath{stroke}
    }
}
\pgfdeclareshape{state}
{
    \savedanchor\centerpoint
    {
        \pgf@x=0pt
        \pgf@y=0pt
    }
    \anchor{center}{\centerpoint}
    \anchorborder{\centerpoint}
    \saveddimen\overallwidth
    {
        \pgf@x=3\wd\pgfnodeparttextbox
        \ifdim\pgf@x<\minimumstatewidth
            \pgf@x=\minimumstatewidth
        \fi
    }
    \savedanchor{\upperrightcorner}
    {
        \pgf@x=.5\wd\pgfnodeparttextbox
        \pgf@y=.5\ht\pgfnodeparttextbox
        \advance\pgf@y by -.5\dp\pgfnodeparttextbox
    }
    \anchor{A}
    {
        \pgf@x=-\overallwidth
        \divide\pgf@x by 4
        \pgf@y=0pt
    }
    \anchor{B}
    {
        \pgf@x=\overallwidth
        \divide\pgf@x by 4
        \pgf@y=0pt
    }
    \anchor{text}
    {
        \upperrightcorner
        \pgf@x=-\pgf@x
        \ifhflip
            \pgf@y=-\pgf@y
            \advance\pgf@y by 0.4\stateheight
        \else
            \pgf@y=-\pgf@y
            \advance\pgf@y by -0.4\stateheight
        \fi
    }
    \backgroundpath
    {
        \pgfsetstrokecolor{black}
        \pgfsetlinewidth{\thickness}
        \pgfpathmoveto{\pgfpoint{-0.5*\overallwidth}{0}}
        \pgfpathlineto{\pgfpoint{0.5*\overallwidth}{0}}
        \ifhflip
            \pgfpathlineto{\pgfpoint{0}{\stateheight}}
        \else
            \pgfpathlineto{\pgfpoint{0}{-\stateheight}}
        \fi
        \pgfpathclose
        \ifblack
            \pgfsetfillcolor{black}
            \pgfusepath{fill,stroke}
        \else
            \pgfusepath{stroke}
        \fi
    }
}
\makeatother

\newcommand{\tinyhandle}[1][dot]{\raisebox{-2pt}{\ensuremath{\hspace{-3pt}\begin{pic}[scale=0.4,string]
        \node (0) at (0,0) {};
        \node[dot, inner sep=1.0pt] (1) at (0,0.3) {};
        \node[dot, inner sep=1.0pt] (2) at (0,0.7) {};
        \node (3) at (0,1) {};
        \draw (0.center) to (1.center);
        \draw (2.center) to (3.center);
        \draw[in=180, out=180, looseness=2] (1.center) to (2.center);
        \draw[in=0, out=0, looseness=2] (1.center) to (2.center);
\end{pic}\hspace{-1pt}}}}

\renewcommand\dag{\ensuremath{\dagger}}

\begin{document}

\title{Orthogonality for Quantum Latin Isometry Squares}
\author{\begin{tabular}{c@{\hspace{50pt}}c}
Benjamin Musto\thanks{Department of Computer Science, University of Oxford, UK} & Jamie Vicary\footnotemark[1]{}\, \thanks{School of Computer Science, University of Birmingham, UK}
\\
\texttt{benjamin.musto@cs.ox.ac.uk}
&
\texttt{j.o.vicary@bham.ac.uk}
\end{tabular}}
\date{\today}
\def\titlerunning{Orthogonality for Quantum Latin Isometry Squares}
\def\authorrunning{B. Musto \& J. Vicary}

\maketitle

\begin{abstract}
Goyeneche et al recently proposed a notion of orthogonality for quantum Latin squares, and showed that orthogonal quantum Latin squares yield quantum codes. We give a simplified characterization of orthogonality for quantum Latin squares, which we show is equivalent to the existing notion. We use this simplified characterization to give an upper bound for the number of mutually orthogonal quantum Latin squares of a given size, and to give the first examples of orthogonal quantum Latin squares that do not arise from ordinary Latin squares. We then discuss quantum Latin isometry squares, generalizations of quantum Latin squares recently introduced by Benoist and Nechita, and define a new orthogonality property for these objects, showing that it also allows the construction of quantum codes. We give a new characterization of unitary error bases using these structures.


\end{abstract}

\section{Introduction}

\subsection{Summary}

At QPL 2016 the present authors introduced \textit{quantum Latin squares}~\cite{Musto2015, Musto2016}, as quantum structures generalizing the well-known Latin squares from classical combinatorics~\cite{latinsquare}. Since then this work has been built on separately by a number of researchers: in particular, by Goyeneche, Raissi, Di Martino and \.Zyczkowski~\cite{Goyeneche2017}, who propose a notion of \textit{orthogonality} for quantum Latin squares which allows the construction of quantum codes; and also by Benoist and Nechita~\cite{Benoist2017}, who introduce \textit{matrices of partial isometries of type (C1,C2,C3,C4)}, generalizations of quantum Latin squares which characterize system-environment observables preserving a certain set of pointer states.

In this paper we give a new formulation of orthogonality for quantum Latin squares, and use it to relate and generalize the works just cited, and extend them in certain ways. In particular, we highlight the following key contributions.
\begin{itemize}
\item We give a new, simplified definition of orthogonality for quantum Latin squares, and show that it is equivalent to the existing definition of Goyeneche et al~\cite[Definition 3]{Goyeneche2017}. (Definition~\ref{def:orthogonalqls}, Theorem~\ref{thm:equivalent}.)
\item We give the first example of a pair of orthogonal quantum Latin squares which are not equivalent to a pair of classical Latin squares. (Example~\ref{example:orthogonal} and Proposition~\ref{prop:notequiv}.) 
\item We show that there can be at most $n-1$ mutually orthogonal quantum Latin squares of dimension~$n$ (Theorem~\ref{thm:upperbound}.)
\item We introduce \textit{quantum Latin isometry squares} based on the  {\textit{matrices of partial isometries of type (C1,C2,C3,C4)}} defined by Benoist and Nechita~\cite[Definition 3.2]{Benoist2017}, and define a new notion of orthogonality for these objects. (Definitions~\ref{def:ipm} and~\ref{def:oipm}.)
\item We show how  orthogonal quantum Latin isometry squares can be used to build quantum codes. (Theorem~\ref{thm:ilstensor}.)
\item We show that unitary error bases give rise to orthogonal pairs of quantum Latin isometry squares, and in fact can be characterized in terms of them. (Theorem~\ref{thm:uebviaorthog}.)
\end{itemize}

{
\subsection{Related work}
Since the introduction of quantum Latin squares by the present authors~\cite{Musto2015}, two notions of orthogonality for quantum Latin squares have been introduced, both of which extend the standard notion for classical Latin squares.

The first such notion, to which we refer here as \emph{left orthogonality}, was introduced by the first author~\cite{Musto2016}, who showed it could be used to construct maximally entangled mutually unbiased bases. Given a pair of classical Latin squares which are left orthogonal by this definition, the left conjugates of each square are orthogonal Latin squares in the traditional sense~\cite{Mann1942}. This notion of orthogonality between QLS is not comparable to that which we study in this paper.

More recently, Goyeneche et al~\cite{Goyeneche2017} introduced another notion of orthogonality for quantum Latin squares, which also extends the traditional definition for classical Latin squares, and the definition which we study here is equivalent.     They extended their notion to quantum orthogonal arrays, more general objects which we do not consider here.
}
\begin{remark}\label{rem:entqls}
Note that the definition of mutually orthogonal quantum Latin squares  introduced by Goyeneche et al includes not only families of quantum  Latin squares satisfying an orthogonality condition but also entangled multi-partite states known as \textit{essentially quantum Latin squares}. This more general definition is exactly equivalent to that of \textit{perfect tensors} which are already well studied. We only consider orthogonality between two or more quantum Latin squares in this paper, and only use the terms orthogonal quantum Latin squares and mutually orthogonal quantum Latin squares to refer to such objects.
\end{remark} 
\subsection{Outline}

This paper has the following structure. In Section 2, we give background on quantum Latin squares, introduce our new definition of orthogonality, and explore its consequences, especially in relation to the work of Goyeneche et al~\cite{Goyeneche2017}. In Section 3, we define quantum Latin isometry squares based on the work of Benoist and Nechita~\cite{Benoist2017}, and investigate a new notion of orthogonality for these objects as well as a connection to unitary error bases.

\ignore{
\BMcomm{Cite our paper introducing QLS and a few others. Mention that a first attempt at orthogonality was introduced in my paper etc.}
\begin{definition}[Quantum Latin square]~\cite{Musto2015}
An $n$-by-$n$ array of unit vectors of the Hilbert space $\C^n$  denoted by $\ket{\psi_{ij}}$ for $i,j\in \{0,...,n-1\}$ is a quantum Latin square if every row and every column are an orthonormal basis. Equivalently if the following equations hold:
\begin{align}
\sum_{i=0}^{n-1}\braket{\psi_{ij}}{\psi_{ik}}\ketbra{ij}{ik}=\sum_{j=0}^{n-1}\braket{\psi_{ij}}{\psi_{kj}}\ketbra{ij}{kj}= \sum_{i,j=0}^{n-1}\ketbra{\psi_{ij}}{\psi_{ij}}\otimes \ketbra{j}{j}=\sum_{i,j=0}^{n-1}\ketbra{\psi_{ij}}{\psi_{ij}}\otimes \ketbra{i}{i}=\I_{n^2}
\end{align}
\end{definition}
}

\section{Quantum Latin squares and orthogonality}

In this section we prove our  main results concerning orthogonal quantum Latin squares. In Section~\ref{sec:basicproperties} we recall the definition of quantum Latin squares, give our new definition of orthogonality, and give a nontrivial example. In Section~\ref{sec:theirorthog} we show that our notion of orthogonality is equivalent to a previous, more complicated definition due to Goyeneche et al~\cite{Goyeneche2017}. In Section~\ref{sec:equivalence} we explore the connection between equivalence and orthogonality of quantum Latin squares, and show that our example of orthogonal quantum Latin squares is not equivalent to a pair of orthogonal classical Latin squares. In Section~\ref{sec:multiplesystems}, we give a simpler definition of orthogonality for families of quantum Latin squares, and show it agrees with that due to Goyeneche et al. In Section~\ref{sec:upperbounds}, we prove an upper bound on the number of mutually orthogonal quantum Latin squares that can exist in any dimension.

\subsection{First definitions}
\label{sec:basicproperties}

We begin with the definition of a quantum Latin square, recently proposed by the present authors~\cite{Musto2015}.

\begin{definition}
A \textit{quantum Latin square (QLS)} $\Psi$ of dimension $n$ is an $n$-by-$n$ array of elements $\ket{\Psi_{ij}} \in \C^n$, such that every row and every column gives an orthonormal basis for $\C^n$. 
\end{definition}

\noindent
A point of notation: when we write $\ket{\Psi_{ij}}$, the indices $i$ and $j$ refer to the row and columns of the array respectively, and take values in the set $[n] = \{0, 1, \ldots, n-1 \}$.

\begin{example}
\label{ex:qls}
Here is a quantum Latin square of dimension 4, given in terms of the computational basis elements \mbox{$\{\ket 0, \ket 1, \ket 2, \ket 3\} \subset \C ^4$:}
{\small\begin{equation*}
\grid{\ket{0} & \ket{1} & \ket{2} & \ket{3}
\\\hline
\frac{1}{\sqrt{2}}(\ket{1}-\ket{2})
& \frac{1}{\sqrt{5}}(i\ket{0}+2\ket{3})
& \frac{1}{\sqrt{5}}(2\ket{0}+i\ket{3})
& \frac{1}{\sqrt{2}}(\ket{1}+\ket{2})
\\\hline
\frac{1}{\sqrt{2}}(\ket{1}+\ket{2})
& \frac{1}{\sqrt{5}}(2\ket{0}+i\ket{3})
& \frac{1}{\sqrt{5}}(i\ket{0}+2\ket{3})
& \frac{1}{\sqrt{2}}(\ket{1}-\ket{2})
\\\hline
\ket{3} & \ket{2} & \ket{1} & \ket{0}}
\end{equation*}}
\end{example}

\noindent
It can readily be checked that along each row, and along each column, the elements form an orthonormal basis for $\C^4$. A \textit{classical Latin square} is a quantum Latin square for which every element of the array is in the computational basis. It is easy to see that classical Latin squares are exactly the ordinary Latin squares studied in combinatorics~\cite{latinsquare}, and so the theory of quantum Latin squares extends this classical theory.

There is a standard notion of orthogonality for classical Latin squares~\cite{Mann1942}. The focus of this paper is the extension of this property to quantum Latin squares, by way of the following new definition.

\begin{definition}
\label{def:orthogonalqls}
Two quantum Latin squares $\Phi, \Psi$ of dimension $n$ are \textit{orthogonal} just when the set of vectors $\{ \ket{\Phi_{ij}} \otimes \ket{\Psi_{ij}} | i,j \in [n]\}$ form an orthonormal basis of the space $\C^n \otimes \C^n$.
\end{definition}

\noindent
We show in Theorem~\ref{thm:equivalent} that this agrees with a more complicated definition recently proposed by Goyeneche et al~\cite{Goyeneche2017}, in terms of partial traces of a tensor expression.

In that paper, it was shown that for classical Latin squares, this agrees with the classical notion of orthogonality. However, no non-classical examples were given of pairs of  quantum Latin squares that are orthogonal . We now rectify this.

\begin{example}[Non-classical orthogonal quantum Latin squares]
\label{example:orthogonal}
Define the unitary matrix $U$ as follows: \begin{equation*}U:=\frac{1}{\sqrt{3}}
\small
\begin{pmatrix}
1 & 1 & 1 & 0 & 0 & 0 & 0 & 0 & 0 \\
1 & e^{\frac{2 \pi i}{3}} & e^{\frac{-2 \pi i}{3}} & 0 & 0 & 0 & 0 & 0 & 0 \\ 
1 & e^{\frac{-2 \pi i}{3}} & e^{\frac{2 \pi i}{3}} & 0 & 0 & 0 & 0 & 0 & 0 \\
0 & 0 & 0 & 1+i & (1-i)/\sqrt{2} & 0 & 0 & 0 & 0   \\
0 & 0 & 0 & -i/\sqrt{2} & 1 & 1/\sqrt{2}+i & 0 & 0 & 0\\
0 & 0 & 0 & 1/\sqrt{2} & i & 1-i/\sqrt{2} & 0 & 0 & 0 \\
0 & 0 & 0 & 0 & 0 & 0 & \sqrt{3/2} & \sqrt{3/2} & 0 \\
0 & 0 & 0 & 0 & 0 & 0 & \sqrt{3/2} & -\sqrt{3/2} & 0\\
0 & 0 & 0 & 0 & 0 & 0 & 0 & 0 & \sqrt{3} \\
\end{pmatrix}\end{equation*}
Then the following arrays are a pair of orthogonal quantum Latin squares of dimension 9:
\begin{align}
 \small\arraycolsep=3pt \grid{ \ket{0}  &  \ket{2}  &  \ket{1}  &  \ket{3}  &  \ket{5}  &  \ket{4}  &  \ket {6}  &  \ket {8}  &  \ket 7 \\ \hline \ket 2 & \ket{1} & \ket{0} & \ket 5  &  \ket 4  &  \ket 3  &  \ket 8  &  \ket 7  &  \ket 6 \\ \hline \ket{1} &  \ket{0} & \ket{2} & \ket{4}  &  \ket{3}  &  \ket 5  &  \ket 7  &  \ket 6  &  \ket 8  \\ \hline \ket{6}  &  \ket{8}  &  \ket{7}  &  \ket{0}  &  \ket 2  &  \ket 1  &  \ket 3  &  \ket 5  &  \ket 4 \\ \hline \ket 8 & \ket 7 & \ket 6 & \ket 2 & \ket 1 & \ket 0 & \ket 5 & \ket 4 & \ket 3 \\ \hline \ket 7 & \ket 6 & \ket 8 & \ket 1 & \ket 0 & \ket 2 & \ket 4 & \ket 3 & \ket 5  \\ \hline U \ket{3} & U\ket{5} & U\ket{4} & \ket 6 & \ket 8 & \ket 7 &  U\ket{0}  & U\ket{2} & U\ket{1}  \\ \hline U\ket{5} & U\ket{4}& U\ket{3} &\ket 8&\ket 7&\ket 6 & U\ket{2} & U\ket{1}& U\ket{0} \\ \hline U\ket{4} & U\ket{3}& U\ket{5} &\ket 7&\ket 6&\ket 8 & U\ket{1}& U\ket{0}&U\ket{2}} & \hspace{10pt} \small\arraycolsep=3pt\grid{ \ket{0}  &  \ket{2}  &  \ket{1}  &  \ket{3}  &  \ket{5}  &  \ket{4}  &  \ket {6}  &  \ket {8}  &  \ket 7 \\ \hline \ket 1 &  \ket{0} &  \ket{2} &  \ket 4  &  \ket 3  &  \ket 5  &  \ket 7  &  \ket 6  &  \ket 8 \\ \hline \ket{2} &  \ket{1} &  \ket{0} &  \ket{5}  &  \ket{4}  &  \ket 3  &  \ket 8  &  \ket 7  &  \ket 6  \\ \hline \ket 3 & \ket 5 & \ket 4 & \ket 6 & \ket 8 & \ket 7 &  U\ket{0}  & U\ket{2} & U\ket{1}  \\ \hline \ket 4 & \ket 3&\ket 5 &\ket 7&\ket 6&\ket 8&U\ket{1}&U\ket{0}&U\ket{2} \\ \hline \ket 5 & \ket 4&\ket 3 &\ket 8&\ket 7&\ket 6&U\ket{2}&U\ket{1}&U\ket{0}\\ \hline U\ket{6}  &  U\ket{8}  &  U\ket{7}  &  \ket{0}  &  \ket 2  &  \ket 1  &  \ket 3  &  \ket 5  &  \ket 4 \\ \hline U\ket{7} & U\ket{6} & U\ket{8} & \ket 1 & \ket 0 & \ket 2 & \ket 4 & \ket 3 & \ket 5 \\ \hline U\ket{8} & U\ket{7} & U\ket{6} & \ket 2 & \ket 1 & \ket 0 & \ket 5 & \ket 4 & \ket 3  }
\end{align}
\end{example}


We now consider some equivalent characterizations of orthogonality, which will be useful later.
\begin{lemma}
\label{lem:conditions}
Two quantum Latin squares $\Phi$, $\Psi$ are orthogonal if and only if one, and hence both, of the following equivalent conditions hold:
\begin{align}
\label{eq:oqlsalg}
\sum_{i,j=0}^{n-1}\ketbra{\Phi_{ij}}{\Phi_{ij}}\otimes\ketbra{\Psi_{ij}}{\Psi_{ij}}=\I_{n^2}
\\
\label{eq:oqlsalg2}
\sum_{i,j,p,q=0}^{n-1}\braket{\Phi_{ij}}{\Phi_{pq}}\braket{\Psi_{ij}}{\Psi_{pq}}\ketbra{ij}{pq}=\I_{n^2}
\end{align}
\end{lemma}
\begin{proof}
For the first condition, equation \eqref{eq:oqlsalg} says that if we sum up outer products of each element of the family $\{ \ket {\Phi_{ij}} \otimes \ket {\Psi_{ij}} | i,j \in [n]\}$, we get the identity; clearly this is equivalent to the statement that the family yields an orthonormal basis. For the second condition, consider the linear map $S= \sum _{i,j} \ket{ij} \bra {\Phi_{ij}} \bra {\Psi_{ij}}$, an operator on $\C^n \otimes \C^n$. The quantum Latin squares $\Phi,\Psi$ are orthogonal if and only if this map is unitary, since it transports the orthonormal basis $\{ \ket{ \Phi_{ij}} \otimes \ket {\Psi_{ij}} | i,j \in [n]\}$ to the computational basis. Since it is an operator on a finite-dimensional Hilbert space, $S$ is unitary if and only if it is an isometry, and equation~\eqref{eq:oqlsalg2} is the isometry condition.
\end{proof}

Orthogonality of quantum Latin squares is unaffected by conjugation of one of the squares.
\begin{definition}
Given a quantum Latin square $\Psi$, its \textit{conjugate} $\Psi^*$ is the quantum Latin square with entries $(\Psi ^*)_{ij} = (\Psi_{ij})^*$.
\end{definition}

\begin{lemma}\label{lem:qlsconj}
Two quantum Latin squares $\Phi, \Psi$ are orthogonal just when $\Phi^*, \Psi$ are orthogonal.
\end{lemma}
\begin{proof}
Suppose $\Phi, \Psi$ are orthogonal quantum Latin squares. Then by equation~\eqref{eq:oqlsalg2},  it follows that $\sum_{i,j,p,q=0}^{n-1}\braket{\Phi_{ij}}{\Phi_{pq}}\braket{\Psi_{ij}}{\Psi_{pq}}=\delta_{ip}\delta_{jq}$. So for all $(i,j)\neq (p,q)$ either $\braket{\Phi_{ij}}{\Phi_{pq}}=0$ or $\braket{\Psi_{ij}}{\Psi_{pq}}=0$, and we know that $\braket{\Phi_{ij}}{\Phi_{ij}}=\braket{\Psi_{ij}}{\Psi_{ij}}=1$. Since $0,1\in \mathbb R$, we conclude that $\sum_{i,j,p,q=0}^{n-1}\braket{\Phi^*_{ij}}{\Phi^*_{pq}}\braket{\Psi_{ij}}{\Psi_{pq}}=\delta_{ip}\delta_{jq}$, and hence by equation~\eqref{eq:oqlsalg} it follows that $\Phi^*, \Psi$ are orthogonal. The converse then follows since $(\Phi^*)^*=\Phi$.
\end{proof}

\ignore{
We have the following graphical characterisation of orthogonal quantum Latin squares.
\begin{lemma}\label{lem:oqlsdiagram}
A pair of QLS are orthogonal if and only if the following linear map is an isometry:
\begin{align}\label{lem:oqlsdiag}
\text{[INSERT PIC OF SKEW PROJ QLS PAIR WITH LEGS BENT - DIAG1]}
\end{align}
\end{lemma}
\begin{proof}
The linear operator~\eqref{lem:oqlsdiag} translates algebraically as $\sum_{i,j=0}^{n-1}\ket{ij}\bra{\Phi_{ij}} (\ket{\Psi_{ij}})^T=\sum_{i,j=0}^{n-1}\ket{ij}\bra{\Phi_{ij}}\otimes \bra{\Psi^*_{ij}}$. By Lemma~\ref{lem:qlsconj} this linear operator is an isometry if and only if the operator $S$ from Remark~\ref{rem:oqlsunitary} is an isometry. 
}
\ignore{By inserting a resolution of the identity in equation~\eqref{eq:oqlsalg} we obtain the following:
\begin{align}
\sum_{i,j,p,q,k}\braket{\Phi_{pq}}{k}\braket{k}{\Phi_{ij}}\braket{\Psi_{ij}}{\Psi_{pq}}\ketbra{ij}{pq}&=\I_{n^2}
\\ \Leftrightarrow \text{Tr}_{\Phi}\Big(\sum_{i,j,p,q}\ket{\Phi_{ij}}_\Phi\bra{\Psi_{ij}}\circ\ketbra{\Psi_{pq}}{\Phi_{pq}}_{\Phi}\otimes\ketbra{ij}{pq}\Big)&=\I_{n^2}
\end{align}
Where Tr$_\Phi$ is the partial trace over the unit vectors of the QLS $\Phi$.
Equation~\eqref{} translates graphically as follows:

INSERT PIC - DIAG2

The result follows by the fact that  isometries between finite dimensional Hilbert spaces of the same dimension are always unitary.
\end{proof}
}

\subsection{Relationship to previous notion of orthogonality}
\label{sec:theirorthog}

The following definition of orthogonality for quantum Latin squares has recently been proposed. It is less conceptual than our Definition~\ref{def:orthogonalqls}, and more complex to work with.
\begin{definition}[\cite{Goyeneche2017}, Definition 3]\label{def:grmzortho}
Two quantum Latin squares $\Phi, \Psi$ are \textit{GRMZ-orthogonal}  when for each tensor factor $X \in \{A, B,C\}$, the following holds:
\begin{equation}\label{eq:theirsingleoqls}
\text{Tr}_X \left(\,\sum_{i,p,j=0}^{n-1} \ket {\Phi_{ij}} \bra {\Phi_{ij}}_A \otimes \ket {\Psi_{pj}} \bra {\Psi_{pj}}_B \otimes \ketbra{ i}{p}_C\right)=\mathbb{I}_{n^2}
\end{equation}
\end{definition}

\noindent
Note that in the presentation of this definition we have restricted the original definition to orthogonality between pairs of quantum Latin squares (please refer to Remark~\ref{rem:entqls}). We now show that Definition~\ref{def:grmzortho} is equivalent to our Definition~\ref{def:orthogonalqls}.

\begin{theorem}
\label{thm:equivalent}
Two quantum Latin squares $\Phi, \Psi$ are orthogonal if and only if they are GRMZ-orthogonal.
\end{theorem}
\begin{proof}
We first consider equation~\eqref{eq:theirsingleoqls} for the case $X=C$, which yields the following equation: 
\begin{equation}
\sum_{i,j,p=0}^{n-1}\ketbra{\Phi_{ij}}{\Phi_{pj}}\otimes \ketbra{\Psi_{ij}}{\Psi_{pj}}\braket{p}{i}=\sum_{i,j=0}^{n-1}\ketbra{\Phi_{ij}}{\Phi_{ij}}\otimes \ketbra{\Psi_{ij}}{\Psi_{ij}}\braket{ij}{ij}=\I_{n^2}
\end{equation}
This corresponds to our equation~\eqref{eq:oqlsalg}. By Lemma~\ref{lem:conditions}, it will hold if and only if $\Phi, \Psi$ are orthogonal.

We now show that the trace conditions over $X=A$ and $X=B$ in the GRMZ-orthogonality definition are redundant, in the sense that they hold automatically for all pairs of quantum Latin squares $\Phi, \Psi$, regardless of orthogonality. We analyze the case that $X=B$; the case $X=A$ is similar. The trace condition yields the following equation:
\begin{align*}
\sum_{i,j,p=0}^{n-1} \braket{\Phi_{ij}}{\Phi_{pj}}\ketbra{\Psi_{ij}}{\Psi_{pj}}\otimes \ketbra{i}{p}=\I_{n^2}
\end{align*}
But this equation holds for any pair of quantum Latin squares $\Phi, \Psi$, as follows:
\begin{align*}
&\sum_{i,j,p=0}^{n-1} \braket{\Phi_{ij}}{\Phi_{pj}}\ketbra{\Psi_{ij}}{\Psi_{pj}}\otimes \ketbra{i}{p}
= \sum_{i,j,p=0}^{n-1}\delta_{ip}\ketbra{\Psi_{ij}}{\Psi_{pj}} \otimes \ketbra i p
= \sum_{i,j=0}^{n-1} \ketbra{\Psi_{ij}}{\Psi_{ij}} \otimes \ketbra i i
\\&= \sum _{i=0}^{n-1} \left( \sum _{j=0} ^{n-1}\ketbra{\Psi_{ij}}{\Psi_{ij}} \right) \otimes \ketbra i i
= \sum _{i=0} ^{n-1} \I_n \otimes \ketbra i i
= \I_n \otimes \left ( \sum _{i=0} ^{n-1} \ketbra i i \right)
= \I_n \otimes \I_n
\end{align*}
Here the first equality uses the fact that $\Phi$ is a QLS, the second equality uses the definition of the Kronecker delta function, the third equality rearranges the sum, the fourth equality uses the fact that $\Psi$ is a QLS, and the final equalities are trivial algebraic manipulations. This completes the proof.
\end{proof} 

\subsection{Equivalence and orthogonality}
\label{sec:equivalence}

Two classical Latin squares are said to be equivalent if one can be transformed into the other by permutations of the rows, columns or computational basis state labels. Similarly, there is a notion of equivalence between quantum Latin squares~\cite{Musto2015}, which we now recall.
\begin{definition}
Two quantum Latin squares $\Phi, \Psi$ of dimension $n$ are \textit{equivalent} if there exists some unitary operator $U$ on $\C^n$, family of modulus-1 complex numbers $c_{ij}$, and permutations $\sigma, \tau \in S_n$, such that the following holds for all $i,j \in [n]$:
\begin{equation}
c_{ij}U\ket{\Phi_{\sigma(i), \tau(j)}}=\ket{\Psi_{ij}}
\end{equation} 
\end{definition} 

\noindent
Orthogonality  is preserved by  taking potentially different equivalences of each QLS, as long as the same pair of permutations are used.

\begin{lemma}\label{lem:equivoqls}
Given quantum Latin squares $\Phi, \Psi, \Phi', \Psi'$ of dimension $n$, unitary operators $U,V$ on $\C^n$, families of modulus-1 complex numbers $c_{ij}, d_{ij}$, and permutations $\sigma, \tau \in S_n$ such that
\begin{align}
\ket{\Phi'_{ij}} &:= c_{ij}U\ket{\Phi_{\sigma(i), \tau(j)}}
&
\ket{\Psi'_{ij}} &:= d_{ij} V\ket{\Psi_{\sigma(i), \tau(j)}}
\end{align}
then $\Phi, \Psi$ are orthogonal if and only if $\Phi', \Psi'$ are orthogonal.
\end{lemma}
\begin{proof}
By Lemma~\ref{lem:conditions} we have the following for all $i,j,m,n$:
\begin{align*}
& \braket{\Phi_{mn}}{\Phi_{ij}} \braket{\Psi_{mn}}{\Psi_{ij}} =\delta_{im}\delta_{jn}
\\ & \Leftrightarrow \bra{\Phi_{m'n'}}U^\dag\circ U\ket{\Phi_{i'j'}}\bra{\Psi_{m'n'}}V^\dag\circ V\ket{\Psi_{i'j'}}=\delta_{i'm'}\delta_{j'n'}c_{i'j'}c^*_{i'j'}d_{i'j'}d^*_{i'j'}=\delta_{i'm'}\delta_{j'n'}c_{m'n'}c^*_{i'j'}d_{m'n'}d^*_{i'j'}
\\ & \Leftrightarrow \bra{\Phi_{m'n'}}U^\dag c^*_{m'n'} c_{i'j'}U\ket{\Phi_{i'j'}}\bra{\Psi_{m'n'}}V^\dag d^*_{m'n'} d_{i'j'}V\ket{\Psi_{i'j'}} =\delta_{i'm'}\delta_{j'n'}
\\ & \Leftrightarrow \braket{\Phi'_{mn}}{\Phi'_{ij}}\braket{\Psi'_{mn}}{\Psi'_{ij}} =\delta_{im}\delta_{jn}
\end{align*} 
Where $\sigma(i)=i',\tau(j)=j',\sigma(m)=m'$ and $\tau(n)=n'$.  
\end{proof}

We now show that the pair of orthogonal quantum Latin squares illustrated in Example~\ref{example:orthogonal} are not equivalent to any pair of orthogonal Latin squares.
\begin{proposition}
\label{prop:notequiv}
The orthogonal quantum Latin squares of Example~\ref{example:orthogonal} are not equivalent to a pair of orthogonal classical Latin squares.
\end{proposition}
\begin{proof}
It is enough to show that the left-hand quantum Latin square of Example~\ref{example:orthogonal}, which we call $\Phi$, is not equivalent to a classical Latin square. Clearly no permutation of the rows or columns could transform $\Phi$ into a classical Latin square. Suppose for a contradiction that there exists a unitary operator  $V$ and a set of phases $c_{ij}$ such that $\ket{\eta_{ij}}:=c_{ij}V\ket{\Phi_{ij}}$ are all computational basis elements, and therefore yield a classical Latin square. Then for all $i,j,m,n$, we must have $\braket{\eta_{mn}}{\eta_{ij}}=0$ or $1$. We choose $m=0$, $n=3$, $i=6$ and $j=2$ to obtain $\braket{\eta_{03}}{\eta_{62}}=c^*_{03}c_{62}^{}\bra{\Phi_{03}}V^\dag V \ket{\Phi_{62}}=c^*_{03}c_{62}^{}\braket{\Phi_{03}}{\Phi_{62}}=c^*_{03}c_{62}^{}\bra{3}U\ket{4}=c^*_{03}c_{62}(1+i)/\sqrt{3}$. But since the $c_{ij}$ have modulus 1, this can never equal $0$ or $1$, and the contradiction is established.
\end{proof}

\ignore{
\subsection{Building 2-uniform tensors}
\label{sec:building2uniform}

We use our definition of orthogonal quantum Latin square, together with the techniques of categorical quantum mechanics, to prove that orthogonal quantum Latin squares yield 2\-uniform tensors.
\begin{theorem}
\label{thm:orthogonalperfect}
Given a  pair of orthogonal QLS the following linear map is a  4-valent error correcting tensor.
\[
\text{[INSERT PIC SKEW PROJECTIVE QLS SHADED CALC - DIAG3]}
\]
\end{theorem}
\begin{proof}
Equation~\eqref{} is the diagonal composition of two biunitaries and is thus a biunitary. By Lemma~\ref{lem:oqlsdiagram}, equation~\eqref{} is therefore an error correcting tensor.
\end{proof}
}

\subsection{Generalization to multiple systems}
\label{sec:multiplesystems}

We now extend this definition to families of quantum Latin squares, generalising mutually orthogonal Latin squares. In particular, we show that no essentially new concept is introduced, with the existing pairwise orthogonality property being sufficient.
\begin{definition}[MOQLS]\label{def:moqls}
A family of $m$ quantum Latin squares $\{\Phi^k|k \in [m]\}$ are \textit{mutually orthogonal} if they are pairwise orthogonal. 
\end{definition}
We now present the definition of mutually orthogonal quantum Latin squares due to Goyeneche et al. As usual we only consider the definition with respect to families of quantum Latin squares (see Remark~\ref{rem:entqls}).
\begin{definition}[GRMZ-MOQLS]
A family of $m$ quantum Latin squares $\{\Phi^k|k \in [m]\}$ are \textit{GRMZ--mutually orthogonal} when the following equations hold, where $X$ indicates a partial trace over any of the $m+1$ subsystems:
\begin{equation}\label{eq:theiroqls}
\text{Tr}_X \left(\,\sum_{i,j,p,q=0}^{n-1}\ket{\Phi^0_{ij}}\bra{\Phi^0_{pq}}\otimes\ket{\Phi^1_{ij}}\bra{\Phi^ 1_{pq}}\otimes...\otimes\ket{\Phi^{m-1}_{ij}}\bra{\Phi^{m-1}_{pq}}\otimes\ket{ij}\bra{pq}\right)=\mathbb{I}_{n^2}
\end{equation}
\end{definition}
\noindent This definition is equivalent to Definition~\ref{def:moqls}.
\begin{proposition}[MOQLS\,=\,GRMZ-MOQLS]
A family of $m$ quantum Latin squares $\{\Phi^k|k \in [m]\}$ are mutually orthogonal just when they are GRMZ--mutually orthogonal.
\end{proposition}
\begin{proof}
We label the $k$th QLS system by $A_k$ and the two other systems in equation~\eqref{eq:theiroqls} as $\alpha$ and $\beta$. So $X$ can range over $m$ element subsets of $\{A_0,A_1,...,A_{m-1},\alpha, \beta\}$. We will label such sets by the two elements that are NOT included so for example $(A_g,A_h)=\{A_0,...,A_{g-1},A_{g+1},...,A_{h-1},A_{h+1},...A_{m-1},\alpha,\beta\}$. 

First we show that for all $g$ and $h$, substituting $X=(A_g,A_h)$ into equation~\eqref{eq:theiroqls} reduces to equation~\eqref{eq:oqlsalg} and so by varying $g$ and $h$ we obtain Definition~\ref{def:moqls}. Let $X=(A_g,A_h)$ then we have:
\begin{align*}
&\sum_{i,j,p,q=0}^{n-1}\braket{\Phi^0_{pq}}{\Phi^0_{ij}} \cdots \braket{\Phi^{g-1}_{pq}}{\Phi^{g-1}_{ij}}\ketbra{\Phi^g_{pq}}{\Phi^g_{ij}}\braket{\Phi^{g+1}_{pq}}{\Phi^{g+1}_{ij}} \cdots \braket{\Phi^{g-1}_{pq}}{\Phi^{h-1}_{ij}}\ketbra{\Phi^h_{pq}}{\Phi^h_{ij}}\braket{\Phi^{h+1}_{pq}}{\Phi^{h+1}_{ij}}
\\
& \qquad \cdots \braket{\Phi^{m-1}_{pq}}{\Phi^{m-1}_{ij}}\braket{pq}{ij}=\I_{n^2}
\qquad\Leftrightarrow\qquad \sum_{i,j=0}^{n-1}\ketbra{\Phi^g_{ij}}{\Phi^g_{ij}} \otimes\ketbra{\Phi^h_{ij}}{\Phi^h_{ij}}=\I_{n^2}
\end{align*} 
 Thus by Lemma~\ref{lem:conditions}, equation~\eqref{eq:theiroqls} with $X=(A_g,A_h)$ holds if and only if $\Phi^g$ and $\Phi^h$ are orthogonal.
We now show that Definition~\ref{def:moqls} implies equation~\eqref{eq:theiroqls}  for all other possible values of $X$.

 Since $\sum_{j=0}^{n-1}\braket{\Phi_{ij}}{\Phi_{pj}}=\delta_{ip}$ by the quantum Latin square property, we have that for  all $k$, substituting $X=(A_k,\alpha)$ into equation~\eqref{eq:theiroqls} reduces to $ \sum_{i,j=0}^{n-1}\ketbra{\Phi^k_{ij}}{\Phi^k_{ij}}\otimes \ketbra{j}{j}=\I_{n^2}$, which holds for all QLS. Similarly by setting $X=(A_k,\beta)$ we obtain $ \sum_{i,j=0}^{n-1}\ketbra{\Phi^k_{ij}}{\Phi^k_{ij}}\otimes \ketbra{j}{j}=\I_{n^2}$ which again holds for all QLS. Finally we are left with $X=(\alpha,\beta)$, which gives the following:
\begin{align}\label{eq:abcase}
\sum_{i,j,p,q=0}^{n-1}\braket{\Phi_{ij}^0}{\Phi_{pq}^0}...\braket{\Phi_{ij}^{m-1}}{\Phi_{pq}^{m-1}}\ketbra{ij}{pq}=\I_{n^2}
\end{align}
Split the $m$ QLS into pairs. Equation~\eqref{eq:oqlsalg2} is equivalent to $\sum_{i,j,p,q=0}^{n-1}\braket{\Phi_{ij}}{\Phi_{pq}}\braket{\Psi_{ij}}{\Psi_{pq}}=\delta_{ip}\delta_{qj}$. If $m$ is even then the LHS of equation~\eqref{eq:abcase} becomes $\sum_{i,j,p,q=0}^{n-1}\delta_{ip}\delta_{pq}\ketbra{ij}{pq}$ which is a resolution of the identity. For $m$ odd we have $\sum_{i,j}^{n-1}\braket{\Phi^{m-1}_{ij}}{\Phi^{m-1}_{ij}}\ketbra{ij}{ij}$ which again is a resolution of the identity since all entries of a QLS are unit vectors. 
\end{proof}


It follows as a corollary of Lemma~\ref{lem:equivoqls}  that MOQLS are preserved by equivalences in the same way as pairs of orthogonal QLS.
\begin{corollary}\label{lem:moqlsunitr}
Given a set of  MOQLS $\Phi^k$, the set of quantum Latin squares with entries $c_{ij}U_k\ket{\Phi_{\sigma(i),\tau(j)}^k}$ are also mutually orthogonal, for any set of unitary operators $U_k$, complex phases $c_{ij}$ and permutations~$\sigma, \tau$.
\end{corollary}  

\subsection{Upper bounds on the number of mutually orthogonal quantum Latin squares}
\label{sec:upperbounds}

We now show that the upper bound for the number of MOQLS of a given size is equal to the upper bound for MOLS. 
\begin{theorem}
\label{thm:upperbound}
Any family of MOQLS of dimension $n$ has size at most $n-1$.
\end{theorem}
\begin{proof}
Suppose that we have a set of $m$-MOQLS $\ket{\Phi^0_{ij}},...,\ket{\Phi^{m-1}_{ij}}$ of size $n\times n$. By Corollary~\ref{lem:moqlsunitr} we can apply unitaries to each QLS such that the first row of every QLS is the ordered computational basis $\ket i , i\in [n]$ so $\ket{\Phi^k_{0i}}=\ket i$ for all $k\in [m],i\in [n]$.  Consider $\braket{\Phi^k_{10}}{\Phi^l_{10}}$ for some $k,l\in [m] $ such that $k\neq l$. We have that:
\begin{align*}
\braket{\Phi^k_{10}}{\Phi^l_{10}}&=\sum_{i=0}^{n-1}\braket{\Phi^k_{10}}{i}\braket{i}{\Phi^l_{10}}=\sum_{i=0}^{n-1}\braket{\Phi^k_{10}}{\Phi^k_{0i}}\braket{\Phi^l_{0i}}{\Phi^l_{10}}=\sum_{i=0}^{n-1}\braket{\Phi^k_{10}}{\Phi^k_{0i}}\braket{\Phi^{l*}_{10}}{\Phi^{l*}_{0i}}\\&=\sum_{i,m,n,p,q=0}^{n-1}\braket{\Phi_{10}^k}{\Phi^k_{0i}}\braket{\Phi^{l*}_{10}}{\Phi_{0i}^{l*}}\braket{10}{mn}\braket{pq}{0i}=\sum_i^{n-1}\braket{10}{0i}=\braket{1}{0}=0
\end{align*}
The first equality is a resolution of the identity, the second  holds since $\ket{\Phi^k_{0i}}=\ket{\Phi^l_{0i}}=\ket i$, the third is a straightforward property of inner products, the fourth equality is simple algebraic rearrangement and the fifth equality is due to Lemma~\ref{lem:qlsconj} and equation~\ref{eq:oqlsalg2}.   So the $m$ unit vectors $\ket{\Phi^i_{10}}$ together with $\ket 0$ are $m+1$ linearly independent vectors. Thus $m$ can be at most~$n-1$. 
\end{proof}
\section{Quantum Latin isometry squares and quantum error detecting codes}
In this section we introduce quantum Latin isometry squares  a generalization of quantum Latin squares and use them to construct quantum error detecting codes. In Section~\ref{sec:ils} we give the definition of quantum isometry Latin squares and give a simple example.
In Section~\ref{sec:skew} we compose pairs of compatible quantum isometry Latin squares and thereby recover both \textit{matrices of partial isometries} and \textit{projective permutation matrices}. In Section~\ref{sec:orthoiso} we give orthogonality criteria for pairs and families of quantum isometry Latin squares. In Section~\ref{sec:error} we show how quantum error detecting codes can be constructed from orthogonal pairs of quantum isometry Latin squares. Finally in Section~\ref{sec:ueb} we show that \textit{unitary error bases} can be characterized as quantum isometry Latin squares that are orthogonal to the \textit{identity square}. 
\subsection{Quantum isometry Latin squares}\label{sec:ils} A normalized vector $\ket{\Psi}$ of dimension $n$ is a trivial example of an isometry $\ket{\Psi}:\C \rightarrow \C^n$. We can thus consider   $n$ dimensional QLS as arrays of isometries of this type. This perspective leads to the following definition which generalises QLS.
\begin{definition}[Quantum isometry Latin square]\label{def:ipm}
An $n$-by-$n$ array of isometries $k_{ij}:\C^{a_{ij}}\rightarrow \C^d$ is a \textit{quantum isometry Latin square}, denoted $(k_{ij},a_{ij},d)$ if the following hold for all $i,j,p,q\in\{0,...,n-1\}$:
\begin{align}
k_{ip}^\dag\circ k_{iq}&=\delta_{pq}\I_{a_{ip}}\label{eq:ipm1}
\\ k_{pj}^\dag\circ k_{qj}&=\delta_{pq}\I_{a_{mj}}\label{eq:ipm2}
\\ \ignore{\bigoplus_{p=0}^{n-1}k_{ip}^\dag\circ k_{ip}&=\bigoplus_{m=0}^{n-1}k_{mj}^\dag \circ k_{mj}=\I_b}\sum_{i=0}^{n-1}k_{ij}\circ k^\dag_{ij}&=\sum_{j=0}^{n-1}k_{ij}\circ k^\dag_{ij}=\I_{d}\label{eq:ipm3}
\end{align}
\end{definition}

\ignore{We now show that biunitarity follows from equations~\eqref{eq:ipm1},\eqref{eq:ipm2} and \eqref{eq:ipm3}.
\begin{lemma}\label{lem:ilsbiuni}
quantum isometry Latin squares are biunitaries. 
\end{lemma}
\begin{proof}
Equations~\eqref{eq:ipm3} ensure that diagram~\eqref{} has the correct dimensions to be biunitary. Interpreting equations~\eqref{eq:ipm1} and~\eqref{eq:ipm2} for diagram~\eqref{} gives us biunitary property equations~\eqref{} and~\eqref{}.\BMcomm{Too handwavy?}  
\end{proof}}
\begin{remark}
A quantum Latin square $\Phi$ of size $n\times n$ is a quantum isometry Latin squares such that $a_{ij}=1$ for all $i$ and $j$, and is therefore of the form $(\ket{\Phi_{ij}},1,n)$. \end{remark}
As a first example, we show how to construct quantum Latin isometry squares from arbitrary families of unitaries.
\begin{example}\label{ex:uniils}
For a Hilbert space $\C^n$ equipped with a family of $m$ unitaries $\mathcal U :=\{U_i: \C^n \to \C^n|i \in [m]\}$, we can build a quantum Latin isometry square, denoted $L(\mathcal U)$, of size $m$: \begin{equation*} L(\mathcal U):=( U_i \delta _{ij}, n \delta _{ij}, n )\end{equation*}Such a quantum Latin isometry square $L(\mathcal U)$ is diagonal, with nonzero isometries only on the leading diagonal. It is straightforward to see that equations~\eqref{eq:ipm1},\eqref{eq:ipm2} and \eqref{eq:ipm3} are satisfied.
\end{example}
\ignore{quantum isometry Latin squares have some useful properties including the following which generalises Theorem~\ref{}. 
\begin{theorem}\label{lem:biuils}
Given an $n\times n$ quantum isometry Latin square $(k_{ij},a_{ij},d)$ the following linear maps are unitary:
\begin{align}
K^R:=\sum_{i,j=0}^{n-1}k_{ij}\otimes\ketbra{i}{ij}& \qquad \qquad
K^C:=\sum_{i,j=0}^{n-1}k_{ij}\otimes\ketbra{j}{ij}
\end{align}
\end{theorem}
\begin{proof}
We have that:
\begin{align*}
(K^R)^\dag\circ K^R&=\sum_{i,j,p,q=0}^{n-1}k_{pq}^\dag \circ k_{ij} \otimes \ket{pq}\braket{p}{i}\bra{ij}\\
&=\sum_{i,j,q=0}^{n-1}k_{iq}^\dag \circ k_{ij} \otimes \ketbra{iq}{ij}\\
&\super{\eqref{eq:ipm1}}=\sum_{i,j}^{n-1}\ketbra{ij}{ij}=\I_{n^2}
\end{align*}
\end{proof}}
\subsection{Skew projective permutation matrices}\label{sec:skew}
Pairs of quantum isometry Latin squares that share the same multiset of values $a_{ij}$ can be composed  to form a new structure.
\begin{definition}[Skew projective permutation matrix]\label{def:sppm}
Given a pair of $n$-by-$n$ quantum isometry Latin squares $(k_{ij},a_{ij},b)$ and $(q_{ij},a_{ij},b), $ let  $T_{ij}:=q_{ij}\circ k_{ij}^\dag$. We define the $n$-by-$n$ array of linear operators $T_{ij}:\C^b\rightarrow \C^b$ to be a \textit{skew projective permutation matrix} (skew PPM). 
\end{definition}
We now show that skew projective permutation matrices are precisely the \textit{matrices of partial isometries of type (C1,C2,C3,C4)}   introduced in a recent paper by Benoist and Nechita~\cite{Benoist2017}. These structures were shown to  characterise quantum channels  preserving pointer states. We first require the following definition.
\begin{definition}[Partial isometry~\cite{Halmos1963}] A \textit{partial isometry} is a linear map such that the restriction to the orthogonal complement of its kernel is an isometry. Alternatively, a partial isometry $A$ is a linear map such that $A\circ A^\dag \circ A=A$. The \textit{initial space} of a partial isometry is the orthogonal complement of its kernel. The \textit{final space} is its range.\end{definition}

Two examples of partial isometries are orthogonal projectors and unitaries. We now give Benoist and Nechita's definition.
\begin{definition}[\cite{Benoist2017}, Definition 3.2 conditions (C1) to (C4)]
An $n$-by-$n$ \textit{matrix of partial isometries of type (C1,C2,C3,C4)} and dimension $b$ is an $n$-by-$n$ array of partial isometries $T_{ij}:\C^b\to\C^b$ such that along each row and column the initial and final spaces of the $T_{ij}$ partition $\C^b$.
\end{definition}
Skew PPMs are matrices of partial isometries of type (C1,C2,C3,C4).

\begin{lemma}\label{lem:skew=partiso}
Given a pair $n$-by-$n$ quantum isometry Latin squares $(k_{ij},a_{ij},b)$ and $(q_{ij},a_{ij},b)$, the corresponding skew PPM is a $b$ dimensional $n$-by-$n$ matrix of partial isometries of type (C1,C2,C3,C4). 
\end{lemma}
\begin{proof}Given the pair of isometries $k_{ij},q_{ij}:\C^{a_{ij}} \rightarrow \C^b$ for some $i,j\in[n]$, we  form the following composite linear maps: 
\begin{align*}
K_{ij}:=k_{ij} \circ k_{ij}^\dag
\quad \quad \quad Q_{ij}:=q_{ij} \circ q_{ij}^\dag
\quad \quad \quad T_{ij}:=q_{ij}\circ k_{ij}^\dag
\end{align*}
It is easy to see that $K_{ij}$ and $Q_{ij}$ are orthogonal projectors. The linear map $T_{ij}$ is a partial isometry since $T_{ij}\circ T_{ij}^\dag \circ T_{ij}=q_{ij}\circ k_{ij}^\dag
\circ k_{ij}\circ q_{ij}^\dag\circ q_{ij}\circ k_{ij}^\dag =q_{ij}\circ k_{ij}^\dag=T_{ij} $ with each equality holding either by definition or using the properties of an isometry. It can  easily be checked that the initial and final spaces of each $T_{ij}$ are the spaces projected onto by $K_{ij}$ and $Q_{ij}$ respectively. By equations~\eqref{eq:ipm1},\eqref{eq:ipm2} and~\eqref{eq:ipm3}  the initial and final spaces of the $T_{ij}$ partition $\C^b$ along the rows and columns as required. 
\end{proof} 
Skew PPMs also appear in another,  currently very active area of research. Skew PPMs are a generalisation of \textit{projective permutation matrices} (PPMs). PPMs are square arrays of orthogonal projectors that form a projective POVM on every row and column. PPMs are  skew PPMs coming from  pairs of identical quantum isometry Latin squares. In this case  the partial isometries $T_{ij}$ are all  orthogonal projectors, having the same initial and final spaces. Clearly these projectors form POVMs on every row and column since the spaces they project onto partition the whole Hilbert space along every row and column by Lemma~\ref{lem:skew=partiso}.  

PPMs, also known as magic unitaries and quantum bijections between classical sets have recently appeared in the context of  quantum non-local games~\cite{Abramsky2017,Atserias2016,Musto2017,Musto2018} and the study of compact quantum groups~\cite{Banica2005,Bichon2003,Wang1998}.
\subsection{Orthogonal quantum isometry Latin squares}\label{sec:orthoiso}
We now  extend the definition of  orthogonal quantum Latin squares to quantum isometry Latin squares. 
\begin{definition}[Orthogonal quantum isometry Latin squares]\label{def:oipm}
A pair of quantum isometry Latin squares $(k_{ij},a_{ij},d)$ and $(q_{ij},a_{ij},d)$ are \textit{orthogonal} if the operators $T_{ij}=q_{ij}\circ k_{ij}^\dag$  span the space of operators and for all non-zero $T_{ij}$ we have that Tr$(T^\dag_{ij}\circ T_{ij})=a$ for some $a\in \C$.
We say that the $T_{ij}$ form an orthogonal skew PPM.  \end{definition}

We now give a more algebraic characterisation of orthogonal quantum isometry Latin squares.
\begin{lemma}\label{lem:oipm}
Given a pair of $n \times n$ quantum isometry Latin squares, $K=(k_{ij},a_{ij},d)$ and $Q=(q_{ij},a_{ij},d)$ define the following linear map $S:\C^n\otimes\C^n\rightarrow \C^d\otimes \C^d$:
\begin{equation}\label{eq:oipm}
S:=\sum_{i,j,p,q=0}^{n-1}\sum_{x=0}^{d-1}\ket{ij}\bra x q_{ij}\circ k_{ij}^\dag\otimes\bra{x}
\end{equation}$K$ and $Q$ are \textit{orthogonal} just when $S$ is  an isometry.
\end{lemma}
\begin{proof}
Assuming $K$ and $Q$ are orthogonal, the partial isometries $q_{ij}\circ k_{ij}^\dag$ span the operator space. We therefore have:
\begin{align}
\I_{d^2}&=\,\,\,\,\,\sum_{i,j=0}^{n-1}\,\,\,\,\,\sum_{x,y=0}^{d-1}\ket y \otimes k_{ij}\circ q_{ij}^\dag\ket{y} \bra x q_{ij}\circ k_{ij}^\dag\otimes\bra{x}\\ \label{eq:tired}&=\sum_{i,j,p,q=0}^{n-1}\sum_{x,y=0}^{d-1}\ket y \otimes k_{pq}\circ q_{pq}^\dag\ket{y}\braket{pq}{ij} \bra x q_{ij}\circ k_{ij}^\dag\otimes\bra{x}\\&=\quad S^\dag\circ S
\end{align}
 The other direction follows straightforwardly.
 \ignore{is equivalent to the following linear map is an isometry $\C^{n}\otimes \C^n\rightarrow \C^{b}\otimes \C^b$:
\begin{equation}
\sum_{i,j=0}^{n}\sum_{x=0}^{b}\ket{ij} \otimes \bra x T_{ij} \otimes \bra{x}
\end{equation}}
\end{proof}
\begin{remark}Note that
$\text{Tr}(T_{ij}^\dag\circ T_{pq})=\delta_{ip}\delta_{jq}\text{Tr}(T_{ij}^\dag\circ T_{ij})$.
\end{remark}

\begin{remark}
A PPM can never be orthogonal since the projectors $T_{ij}$ of a PPM span the operator space for every row and column.
\end{remark}
Orthogonal  quantum isometry Latin squares generalise orthogonal QLS (and therefore orthogonal Latin squares).
\begin{lemma}
Pairs of QLS are orthogonal  quantum isometry Latin squares if and only if they are orthogonal quantum Latin squares.
\end{lemma}
\begin{proof}Consider a pair of $n$-by-$n$  QLS $(\ket{k_{ij}},1,n)$ and $(\ket{q_{ij}},1,n)$ such that they are orthogonal quantum isometry Latin squares by Definition~\ref{def:oipm}.
By Lemma~\ref{lem:oipm}, $S$ is an isometry. Since $S$ is a linear operator on a finite dimensional Hilbert space, $S$ is unitary.  This yields the following equation: 
\begin{equation*}\sum_{i,j,x,y=0}^{n-1}\ket{k_{ij}} \braket{q_{ij}}x \braket{y}{q_{ij}}\bra{k_{ij}}\otimes \ketbra{x}{y}=\sum^{n-1}_{i,j=0}\ketbra{k_{ij}}{k_{ij}}\otimes\ketbra{q_{ij}^*}{q_{ij}^*}=\I_{n^2}
\end{equation*} By Lemmas~\ref{lem:conditions} and~\ref{lem:qlsconj} this holds if and only if $\ket{q_{ij}}$ and $\ket{k_{ij}}$ are orthogonal QLS.   
\end{proof}
We define \textit{mutually orthogonal quantum isometry Latin squares,} to be sets of pairwise orthogonal quantum isometry Latin squares thus generalising MOLS and MOQLS (see Definition~\ref{def:moqls}).

We now present an example pair of orhogonal quantum isometry Latin squares which are not QLS.
\begin{example}\label{ex:orthoiqls}
We present a pair of orthogonal quantum Latin isometry squares $Q$ and $K$ and associated orthogonal skew PPM, $T$. We have $n=8$,  $d=4$ and $a_{ij}=2$ or $0$ for all $i,j\in \{0,...,7\}$. There are $d^2=16$ non-zero $T_{ij}$ as required to span the operator space.

We fix the computational basis $\ket a, \ket b$ for $\C^2$ and $\ket 0, \ket 1, \ket 2\ \ket 3$ for $\C^4$.

\noindent We present the first quantum Latin isometry square $Q$:
\begin{align*}
 \footnotesize\arraycolsep=2pt \grid{
 \ketbra{0}{a}+\ketbra{1}{b} &
\ketbra 2 a + \ketbra 3 b &
 0 & 0 & 0 & 0 & 0 & 0 \\
 \hline
\ketbra 2 a + \ketbra 3 b & 
\ketbra{1}{a}+\ketbra{0}{b}&
 0 & 0 & 0 & 0 & 0 & 0 \\
 \hline
 0 & 0 &
 \ketbra{0}{a}-\ketbra{1}{b}&
 \ketbra{3}{b}-\ketbra{2}{a}&
 0&0&0&0 \\
 \hline
 0&0&
 \ketbra{2}{a}-\ketbra{3}{b}&
 \ketbra{0}{a}-\ketbra{1}{b}&
 0&0&0&0\\
 \hline
 0&0&0&0&
 \ketbra{1}{a}+\ketbra{2}{b}&
 \ketbra{0}{a}+\ketbra{3}{b}&
 0&0\\
 \hline
 0&0&0&0&
 \ketbra{3}{b}-\ketbra{0}{a}&
 \ketbra{2}{b}-\ketbra{1}{a}&
 0&0\\
 \hline
 0&0&0&0&0&0&
 \ketbra{0}{a}+\ketbra{2}{b}&
 \ketbra{3}{b}-\ketbra{1}{a}\\
 \hline
 0&0&0&0&0&0&
 \ketbra{1}{a}+\ketbra{3}{b}&
 \ketbra{2}{b}-\ketbra{0}{a}
}\end{align*} 
Now we present the quantum Latin isometry square $K$:
\begin{align*}
\footnotesize\arraycolsep=2pt\grid{
\ketbra{0}{a}+\ketbra{1}{b}&
\ketbra{2}{a}+\ketbra{3}{b}&
0&0&0&0&0&0\\
\hline
\ketbra{2}{a}+\ketbra{3}{b}&
\ketbra{0}{a}+\ketbra{1}{b}&
0&0&0&0&0&0\\
\hline
0&0&
\ketbra{0}{a}+\ketbra{1}{b}&
\ketbra{3}{a}+\ketbra{2}{b}&
0&0&0&0\\
\hline
0&0&
\ketbra{2}{a}+\ketbra{3}{b}&
\ketbra{1}{a}+\ketbra{0}{b}&
0&0&0&0\\
\hline
0&0&0&0&
\ketbra{2}{a}+\ketbra{1}{b}&
\ketbra{3}{a}+\ketbra{0}{b}&
0&0\\
\hline
0&0&0&0&
\ketbra{3}{a}+\ketbra{0}{b}&
\ketbra{2}{a}+\ketbra{1}{b}&
0&0\\
\hline
0&0&0&0&0&0&
\ketbra{2}{a}+\ketbra{0}{b}&
\ketbra{3}{a}+\ketbra{1}{b}\\
\hline
0&0&0&0&0&0&
\ketbra{3}{a}+\ketbra{1}{b}&
\ketbra{2}{a}+\ketbra{0}{b}
}\end{align*}    
Finally we present the associated skew projective permutation matrix $T$ with entries $T_{ij}$:
\begin{align*}
\footnotesize\arraycolsep=2pt\grid{
\ketbra{0}{0}+\ketbra{1}{1}&
\ketbra{3}{2}+\ketbra{2}{3}&
0&0&0&0&0&0\\
\hline
\ketbra{2}{2}+\ketbra{3}{3}&
\ketbra{1}{0}+\ketbra{0}{1}&
0&0&0&0&0&0\\
\hline
0&0&
\ketbra{0}{0}-\ketbra{1}{1}&
\ketbra{3}{2}-\ketbra{2}{3}&
0&0&0&0\\
\hline
0&0&
\ketbra{2}{2}-\ketbra{3}{3}&
\ketbra{0}{1}-\ketbra{1}{0}&
0&0&0&0\\
\hline
0&0&0&0&
\ketbra{2}{1}+\ketbra{1}{2}&
\ketbra{3}{0}+\ketbra{0}{3}&
0&0\\
\hline
0&0&0&0&
\ketbra{3}{0}-\ketbra{0}{3}&
\ketbra{2}{1}-\ketbra{1}{2}&
0&0\\
\hline
0&0&0&0&0&0&
\ketbra{2}{0}+\ketbra{0}{2}&
\ketbra{3}{1}-\ketbra{1}{3}\\
\hline
0&0&0&0&0&0&
\ketbra{3}{1}+\ketbra{1}{3}&
\ketbra{2}{0}-\ketbra{0}{2}
}\end{align*} 

\end{example}
\subsection{Quantum error detecting codes}\label{sec:error}
We now prove the main result of this section, a construction of quantum error detecting codes from  orthogonal skew PPMs.
In their famous 1997 paper Knill and Laflamme proved that certain one-to-three $4$-valent tensors can be used as encoding maps for quantum codes that can detect any single local error. 
\begin{theorem}~\cite{Knill1997}
Given a three-to-one tensor $\bra{E_{ijk}}:\C^a\rightarrow \C^b\otimes \C^c\otimes \C^d$,  it is  an encoding map that detects a single error
if the following hold:
\begin{align}
\sum_{i=0}^{b-1} \sum_{j=0}^{c-1} \sum_{l,k=0}^{d-1}\ketbra{E_{ijl}}{E_{ijk}}\otimes\ketbra{l}{k}&=\I_a \otimes \I_d\label{eq:error1}
\\ \sum_{i=0}^{b-1} \sum_{l,j=0}^{c-1} \sum_{k=0}^{d-1}\ketbra{E_{ilk}}{E_{ijk}}\otimes\ketbra{l}{j}&=\I_a \otimes \I_c\label{eq:error2}\\
\sum_{l,i=0}^{b-1} \sum_{j=0}^{c-1} \sum_{k=0}^{d-1}\ketbra{E_{ljk}}{E_{ijk}}\otimes\ketbra{l}{i}&=\I_a \otimes \I_b\label{eq:error3}
\end{align}
\end{theorem} 
\noindent  
\begin{theorem}
\label{thm:ilstensor}
Given an  $n$-by-$n$ pair of orthogonal quantum isometry Latin squares
$(k_{ij},a_{ij},d)$ and $(q_{ij},a_{ij},d)$; the following one-to-three tensor is an encoding map that detects a single error:
\begin{equation}
\bra{T}:=\sum_{i,j=0}^{n-1}\ket i \otimes q_{ij}\circ k^\dag_{ij}\otimes \ket j
\end{equation}

\end{theorem}
\begin{proof}
First we show equation~\eqref{eq:error1}, we have:
\begin{align*}
\sum_{x=0}^{d-1}\,\,\, \sum_{l,i,j=0}^{n-1}\ketbra{T}{T}\otimes\ketbra{l}{i}&=\sum_{x=0}^{d-1}\,\,\, \sum_{l,i,j=0}^{n-1}\ket l \otimes k_{lj}\circ q^\dag_{lj}\ketbra x x q_{ij}\circ k^\dag_{ij}\otimes \bra i\\
&\super{\eqref{eq:ipm2}}=\sum_{i,j=0}^{n-1}\ket i \otimes k_{ij}\circ k^\dag_{ij}\otimes \bra i\super{\eqref{eq:ipm3}}=\I_d \otimes \I_n  
\end{align*}  
Equation~\eqref{eq:error3} can be derived from equations~\eqref{eq:ipm1} and ~\eqref{eq:ipm3} similarly.

We now show equation~\eqref{eq:error2}.
\begin{align*}
\sum_{x,y=0}^{d-1}\,\,\, \sum_{i,j=0}^{n-1}\ketbra{T}{T}\otimes\ketbra{y}{x}&=\sum_{x,y=0}^{d-1}\,\,\, \sum_{i,j=0}^{n-1}\ket y \otimes k_{ij}\circ q^\dag_{ij}\ket y\bra x q_{ij} \circ k_{ij}^\dag \otimes \bra x\\
&\super{\eqref{eq:tired}}=\I_{d^2}
\end{align*} 
\end{proof}
\begin{example}Given the orthogonal skew PPM $T=\{T_{ij}|i,j\in \{0,...,7\}\}$ as in Example~\ref{ex:orthoiqls}, by Theorem~\ref{thm:ilstensor} the following three-to-one 4-valent tensor  is an encoding map that detects a single qubit local error:
 \begin{equation*}
 \bra T:=\sum_{i,j=0}^{n-1}\ket i \otimes T_{ij}\otimes \ket j
 \end{equation*}
\end{example} 
\subsection{Orthogonal quantum Latin isometry squares from unitary error bases}\label{sec:ueb}

Here we recall the standard notion of unitary error basis, and show that they can be characterized as orthogonal pairs of quantum Latin isometry squares, which are not quantum Latin squares. Unitary error bases were introduced by Werner~\cite{Werner:2001}, and provide the basic data for quantum teleportation, dense coding and error correction procedures~\cite{Werner:2001, Knill:1996:2, Shor:1996}.

\begin{definition}
For a Hilbert space $\C^n$, a \textit{unitary error basis} is a family of unitary operators $U_i:\C^n\ \to \C^n$ which span the space of operators, and which are orthogonal under the trace inner product:
\begin{equation}
\Tr(U_i^{} \circ U_j ^\dag) = n \delta_{ij}
\end{equation}
\end{definition}

We have already seen in Example~\ref{ex:uniils} that any family of unitaries is a quantum Latin isometry square. We now show that unitary error bases can be characterized in terms of an orthogonal pair of quantum Latin isometry squares. We first define the \textit{identity  square}.
\begin{definition}[Identity  square]
The $d$-dimensional \textit{identity  square} is  the quantum Latin isometry square given by $( \I_d \delta _{ij}, d \delta _{ij}, d )$.
Note that the identity square is also a skew PPM and a PPM.\end{definition}
We now present the main result of this section.  
\begin{theorem}
\label{thm:uebviaorthog}
The following are equivalent:
\begin{itemize}
\item the set of all $n^2$-by-$n^2$ quantum Latin isometry squares that are orthogonal to the $n^2$-dimensional identity square;
\item the set of all unitary error basis for $\C^n$.
\end{itemize}
\end{theorem}
\begin{proof}
First note that any quantum Latin isometry square that is orthogonal to the identity square must be of  the form $(X_{ij},d\delta_{ij},d)$ in order to be composed.  The isometries $X_{ij}$ are linear operators and must  therefore be unitary. The orthogonality condition for the quantum Latin isometry squares unpacks to the requirement that the unitary operators $X_i$ are orthogonal and span the space of operators; this is exactly the unitary error basis condition.
\end{proof}  
\setlength{\bibsep}{0pt plus 0.3ex}
\bibliographystyle{plainurl}
\bibliography{PerfectTensors}

\end{document}